\documentclass[prx,twocolumn,floatix,amssymb,amsmath, superscriptaddress,longbibliography]{revtex4-2}
\usepackage{graphicx}
\usepackage{epsfig}
\usepackage{amsmath,amssymb,amsthm}
\usepackage[utf8]{inputenc}
\usepackage[dvipsnames]{xcolor}
\usepackage[colorlinks, urlcolor=blue, linkcolor=red]{hyperref}
\usepackage{soul,xcolor}
\usepackage{soul}
\usepackage{float} 
\setstcolor{red}
\newcommand{\beqa}{\begin{eqnarray}}
\newcommand{\eeqa}{\end{eqnarray}}

\renewcommand{\boxed}[2]{\textcolor{#1}{%
\tikz[baseline={([yshift=-1ex]current bounding box.center)}] \node [rectangle, minimum width=1ex,rounded corners,draw] {\normalcolor\m@th$\displaystyle#2$};}}
\DeclareMathOperator{\Tr}{Tr[]}
\newcounter{appsection}
\newcounter{appsubsection}[appsection]

\usepackage{physics}
\usepackage{bookmark}
\usepackage[normalem]{ulem}

\newcommand\redsout{\bgroup\markoverwith{\textcolor{red}{\rule[0.5ex]{2pt}{2pt}}}\ULon}
\newcommand\bluesout{\bgroup\markoverwith{\textcolor{blue}{\rule[0.5ex]{2pt}{2pt}}}\ULon}
\newcommand\greensout{\bgroup\markoverwith{\textcolor{green}{\rule[0.5ex]{2pt}{2pt}}}\ULon}
\newcommand{\overbar}[1]{\mkern 1.5mu\overline{\mkern-1.1mu#1\mkern-0.3mu}\mkern 1.5mu}

\DeclareMathOperator{\arcsinh}{arcsinh}

\newtheorem*{result*}{Main Result}

\newtheorem*{theo*}{Theorem}

\begin{document}

\title{Hybrid Analog Teleportation–Direct Transmission in Noisy Bosonic Channels}

\author{Uesli Alushi}
\email{uesli.alushi@aalto.fi}
\affiliation{Department of Information and Communications Engineering, Aalto University, Espoo 02150, Finland}
\author{Simone Felicetti}
\email{simone.felicetti@cnr.it}
\affiliation{Institute for Complex Systems, National Research Council (ISC-CNR), Via dei Taurini 19, 00185 Rome, Italy}
\affiliation{Physics Department, Sapienza University, P.le A. Moro 2, 00185 Rome, Italy}
\author{Roberto Di Candia}
\email{rob.dicandia@gmail.com}
\affiliation{Department of Information and Communications Engineering, Aalto University, Espoo 02150, Finland}
\affiliation{Dipartimento di Fisica, Universit\`a degli Studi di Pavia, Via Agostino Bassi 6, I-27100 Pavia, Italy}
\begin{abstract}
Quantum teleportation uses a shared entangled resource, local operations, and a digitally error-corrected classical channel to transfer quantum states between distant parties. We introduce a hybrid teleportation-direct transmission protocol for state transfer that still exploits entanglement, but replaces classical communication and digital error correction with an {\it analog} feedforward through a noisy quantum channel. We show that quantum teleportation outperforms this protocol if the communication channel reduces the entanglement of all bipartite states having the same amount of entanglement as the resource; otherwise, the hybrid protocol is optimal. We apply our result to the state transfer of a uniformly distributed coherent-states codebook, highlighting experimentally relevant scenarios where our protocol is most effective. Our findings are directly relevant to both optical and superconducting microwave channels, where analog feedforward techniques have been recently implemented.
\end{abstract}

\maketitle

\section{Introduction}
Quantum teleportation allows two parties situated in distant locations, Alice and Bob, to
transfer an unknown quantum state using entanglement, local operations, and a classical communication channel as resources~\cite{Bennett1993,Vaidman1994, Braunstein1998}. The protocol consists of four steps: (i) A third party (Charlie) prepares an entangled pair and sends one share to Alice and the other to Bob; (ii) Alice performs a joint (Bell) measurement of the unknown state and her part of the entangled state; (iii) Alice sends the measurement results to Bob via a classical communication channel; (iv) Bob performs a local operation on his part of the entangled state, depending on Alice classical message. Since the early theoretical proposals, experimental implementations of quantum teleportation have been demonstrated on several quantum platforms~\cite{pirandola2015advances,hu2023}. The performance dependence on steps (i)-(ii) and (iv) has already been well studied in the literature: The protocol is limited by the amount of pre-shared entanglement, the quantum memory of Bob, the fidelity of the Bell measurement at Alice, and the noise of the decoding operation at Bob~\cite{olivares2003,pirandola2015advances,hu2023}. Usually, step (iii) is treated as noise-free, as one can perform digital error correction, which means that the classical information can be transmitted virtually without errors. Nevertheless, a fundamental question remains: Is quantum teleportation always the best choice for state transfer, or can alternative strategies surpass it over noisy {\it quantum} channels?

In this paper, we answer this question in a continuous-variable setting. We investigate an analog protocol that leverages the quantum nature of the communication channel by replacing digital error correction in quantum teleportation with an analog error correction scheme. In this approach, the Bell measurement is replaced by an encoding operation implemented via a quantum-limited two-mode squeezer. The encoded output is then transmitted directly to Bob, who applies a decoding operation to retrieve the unknown quantum state, see Fig.~\ref{Fig1}. For {\it infinite} encoding squeezing, the analog protocol is equivalent to quantum teleportation, and we call it ``Analog Quantum Teleportation''. When no encoding is applied, the protocol represents direct state transfer. For {\it finite} encoding squeezing, it is a hybrid teleportation-direct transmission protocol (HTDT).

\begin{figure}[t!]
    \centering
    \includegraphics[width=.45\textwidth]{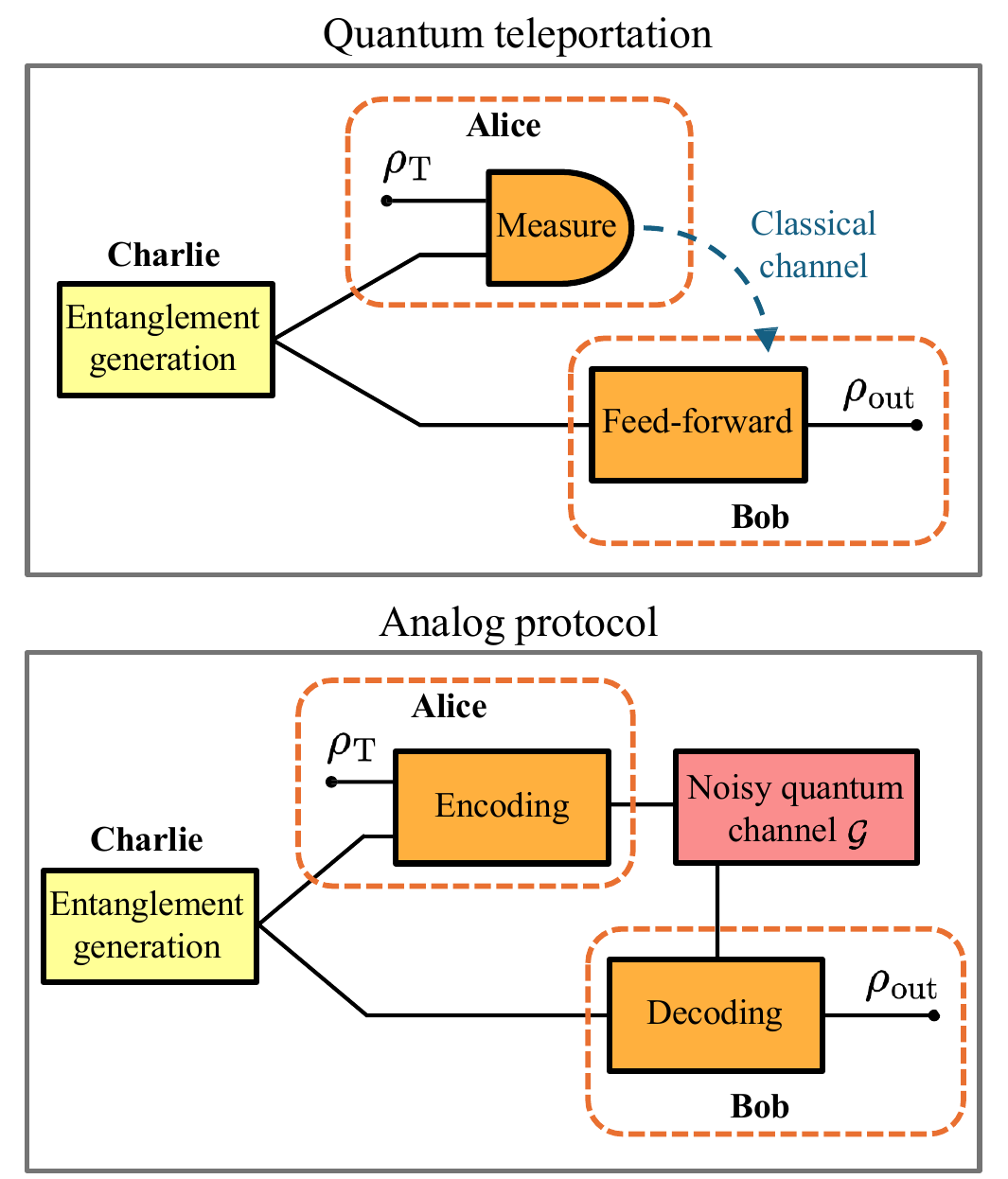}
    \caption{\textbf{Sketch of quantum teleportation and the analog protocol}. Both protocols leverage an entangled resource state, distributed by Charlie, to transfer a state from Alice to Bob obliviously. In the analog protocol, we exploit the quantum description of the communication channel \(\mathcal{G}\) to implement an optimal encoding/decoding scheme, which replaces the measure/feedforward operations of quantum teleportation.}
    \label{Fig1}
\end{figure}

We first frame the problem within the broader context of Gaussian channel simulation, where quantum teleportation corresponds to the case of noisy simulation of a phase-insensitive channel~\cite{Adesso2017,tserkis2018simulation}. Here, the goal is to reproduce the effect of a given channel with minimal added noise. We identify a necessary and sufficient condition for HTDT to be optimal, based solely on the communication channel parameters and the entanglement of the resource, quantified by the logarithmic negativity~\cite{Vidal2002,Plenio2005}.
We showcase the performance of the protocol for a codebook of uniformly distributed coherent states, considering both cases in which the entanglement distribution channel is effectively noiseless compared to the communication channel, and the ``fair'' case where all channels have the same loss rate. In both scenarios, our HTDT protocol can offer a significant advantage over quantum teleportation when entanglement resources are finite.


So far, analog feedforward has mainly been regarded as a technical tool to speed up protocol implementations \cite{Ralph1999, Dicandia2015}, as the absence of measurements avoids analog-to-digital and digital-to-analog conversions and reduces the waiting time at Bob’s side. Recent experiments in the optical~\cite{Liu2020Orbital, Liu2024} and microwave~\cite{Fedorov2021experimental,yam2025quantum,Abdo2025} regimes have demonstrated analog feedforward in the strong-squeezing limit to implement quantum teleportation. In this paper, we study scenarios in which the optimal encoding squeezing is \emph{finite}, corresponding to an HTDT protocol, and provide its optimal value.
This result is relevant for solid-state quantum technologies, such as superconducting quantum circuits connected via cryolinks~\cite{walraff2020, yam2023}, enabling kilometer-scale state transfer and outperforming quantum teleportation for realistic parameters, with direct applications in modular quantum computing~\cite{eisert2000,bravyi2022} and quantum communication~\cite{usenko2025}.

The paper is organized as follows: in Section~\ref{gaussian} we introduce the Gaussian formalism for continuous-variable quantum mechanics; in Section~\ref{analog}, we describe our analog protocol and prove a necessary and sufficient condition for HTDT to outperform quantum teleportation; in Section~\ref{opt_state_main}, we discuss the transfer of a uniformly distributed codebook of coherent states, considering both the case of lossless entanglement distribution and the ``fair'' case in which the entanglement distribution and communication channels have the same loss rate; finally, in Section~\ref{discussion} we conclude.

\section{ Gaussian formalism}\label{gaussian}
An \textit{m}-mode bosonic system~\cite{ AlessioSerafini} is described by the self-adjoint canonical operators \(\hat{x}_j\) and \(\hat{p}_j\), with \({j=1,...,m}\),  satisfying the canonical commutation relations. If we define \({\hat{r}=(\hat{x}_1,\hat{p}_1,...,\hat{x}_m,\hat{p}_m)^T}\), the canonical commutation relations read \({[\hat{r}_j,\hat{r}_k]=i\Omega_{jk}}\),
 where \({\Omega=\bigoplus_{j=1}^{m} i\sigma_y}\) is the symplectic form. Here and in the following, \(\sigma_x\), \(\sigma_y\), \(\sigma_z\) are the Pauli matrices. Gaussian states \(\rho_{\rm G}\) are fully characterized by their first-moments vector $v=\Tr[\rho_{\rm G}\hat{r}]$ and their covariance matrix $\Gamma$ with components \({\Gamma_{jk}=\Tr[\rho_{\rm G}\{\hat{r}_j-v_j,\hat{r}_k-v_k\}]}\). In this representation, the Robertson-Schrödinger uncertainty relation \( {\Gamma+i\Omega\geq0}\) must be satisfied. 

In the following, we consider the entanglement resource to be a two-mode Gaussian state with null first moments and covariance matrix
 \begin{equation}\label{covariance}
     \Gamma_{\rm AB}=
     \begin{pmatrix}
     a\mathbb{I}_2 &-c\sigma_z\\
     -c\sigma_z & b\mathbb{I}_2
     \end{pmatrix}\,,
 \end{equation}
with \({a,b\geq1}\) and \({0\leq c\leq\sqrt{ab-1-\left|a-b\right|}}\). This state can be generated using two-mode squeezing operations applied to uncorrelated thermal states~\cite{menzel2012path,Flurin2012, pogorzalek,Fedorov2021experimental, Abdo2025}. 
The entanglement can be quantified using the logarithmic negativity \({E_\mathcal{N}(\rho_{\rm AB})=\ln\|\rho_{\rm AB}^{T_{\rm A}}\|_1}\)~\cite{Vidal2002,Plenio2005}, where $\rho_{\rm AB}$ is the density matrix of the resource state and $T_{\rm A}$ denotes the partial transpose. For two-mode Gaussian states, this corresponds to \( {E_\mathcal{N}= \max\{0,-\ln{\left(\nu_-\right)}\}}\), 
where \(\nu_-\) is the lowest symplectic eigenvalue of the partial transposed covariance matrix \(\left(\mathbb{I}_2 \oplus \sigma_z\right)\Gamma_{\rm AB}\left(\mathbb{I}_2 \oplus \sigma_z\right)\)~\cite{AlessioSerafini2003}. Other measures of entanglement can be considered to quantify the teleportation resource~\cite{tserkis2018simulation}. However, here we consider the logarithmic negativity as it will give a simple optimality condition for the analog protocol. 

Gaussian-preserving channels are defined by two real matrices $X$ and $Y$ acting on the first-moments vector and covariance matrix as 
\begin{align}\label{Gaussianaction}
v\rightarrow{}Xv\;\;\;,\;\;\; \Gamma\rightarrow{}X\Gamma X^T+Y\,.
\end{align}
The complete positivity condition is \(Y+i\Omega\geq iX\Omega X^T\)~\cite{AlessioSerafini}. 
For single-mode Gaussian channels, this condition reduces to \({\sqrt{\det(Y)}\geq\left|1-\det(X)\right|}\) and \({Y\geq0}\). 

Throughout this work, we consider the case in which Alice and Bob share a Gaussian communication channel \(\mathcal{G}\), which can  be 
decomposed as  
\begin{equation}\label{diagonalization}
    \mathcal{G}[\rho]=W\left(\mathcal{E}\left[U\rho U^\dagger\right]\right)W^\dagger\,.
\end{equation}
Here, \(U\), \(W\) are two unitary matrices and \(\mathcal{E}\) is the Gaussian channel in the canonical form \cite{holevo2007one,Weedbrook2012}. Since the unitaries $U$ and $W$ can be reversed in the encoding and decoding stages, we can restrict the analysis to channels in their canonical form. Indeed, in the following, we focus only on non-degenerate channels, i.e., those channels that do not erase the information carried on one quadrature. Specifically, we consider phase-insensitive channels defined by
\begin{equation}\label{phasechannel}
    X=\sqrt{x}\mathbb{I}_2 \;\;\;,\;\;\;Y=y\mathbb{I}_2\,,
\end{equation}
where \({x>0}\) is the transmissivity and \(y\) is the added noise~\footnote{Referring to Table~I of Ref.~\cite{Weedbrook2012}, we do not consider quantum channels in the classes $A_1$ and $A_2$ as they are degenerate. The class $B_1$ can be treated as a phase-insensitive channel by adding noise in the $\hat{x}$ quadrature. Classes $B_2$ and $\mathcal{C}$ are of the form given in Eq.~\eqref{phasechannel}. Class $D$ acts as a noisy complex conjugation, which can be corrected at the encoding stage in a noisy way.}. The complete positivity condition for this channel is \({y\geq\left|1-x\right|}\). If ${y\geq e^{-2r}(1+x)}$, then applying the channel to a part of any state with logarithmic negativity larger than $2r$ will necessarily reduce the entanglement at the output. A particular case is ${r=0}$, for which ${y\geq 1+x}$ means the channel is entanglement-breaking.  

\section{Analog protocol}\label{analog}

\subsection{Gaussian Channel Simulation} Alice aims to transmit to Bob a quantum state obtained by applying a phase-insensitive Gaussian channel with gain $g$ and noise $G$ to an unknown input state $\rho_{\rm T}$, characterized by first and second moments $v_{\rm T}$ and $\Gamma_{\rm T}$, respectively. She shares with Bob an entangled state \(\rho_{\rm AB}\) with covariance matrix~\eqref{covariance}. The first-moment vector and the covariance matrix of the initial three-mode state are then $v= (v^T_{\rm T},
    0,0,0,0
)^T_{\rm AB}$ and $\Gamma=\Gamma_{\rm T}\oplus\Gamma_{\rm AB}$.
The analog protocol consists of three steps, see Fig.~\ref{Fig2}: (i) Alice applies an encoding operation to \(\rho_{\rm T}\) and her part of the entangled state; (ii) Alice transmits the encoded signal to Bob through the Gaussian communication channel, defined in \eqref{phasechannel}; (iii) Bob applies a decoding operation to retrieve a ``noisy'' version of \(\rho_{\rm T}\). 

\begin{figure}[t!]
    \centering
    \includegraphics[width=.47\textwidth]{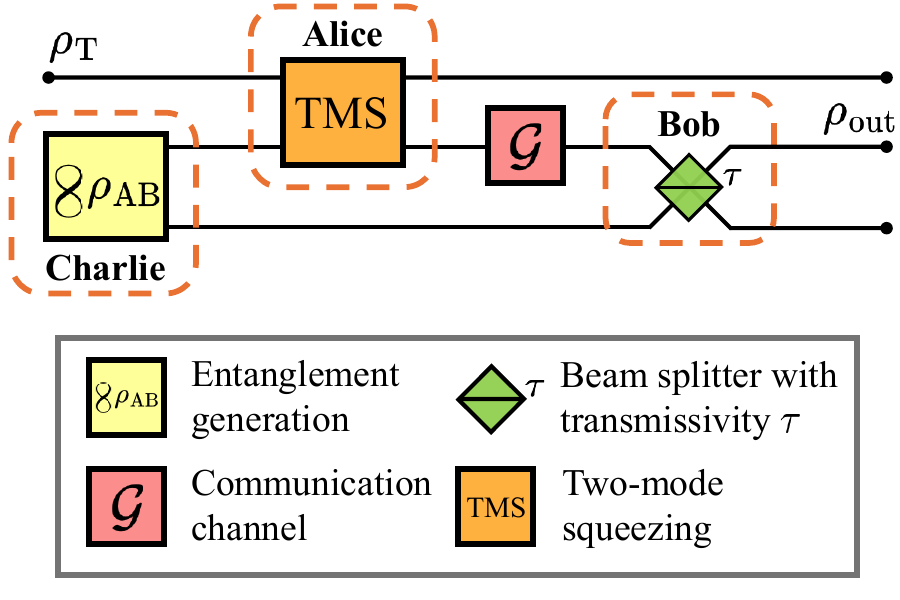}
    \caption{ \textbf{Analog  scheme for quantum state transfer}. Alice aims to transfer an unknown state $\rho_{\rm T}$ to Bob using the pre-shared entangled state $\rho_{\rm AB}$ as a resource. She applies a two-mode squeezing operation, parametrized by the linear gain ${d\geq1}$ (see Eq.~\eqref{encoding}), to her share of $\rho_{\rm AB}$ and the state $\rho_{\rm T}$ (with first moment $v_T$ and covariance matrix $\Gamma_{\rm T}$). She transmits the mode with first moment $\sqrt{d}\,v_{\rm T}$ to Bob through the communication channel $\mathcal{G}$, while discarding the other mode, which corresponds to its complex conjugate.  Finally, Bob performs a decoding operation using his part of the entangled state, obtaining $\rho_{\rm out}$, which is a ``noisy'' version of the state \(\rho_{\rm T}\). For infinite $d$, this is an Analog Quantum Teleportation protocol. For ${d=1}$, it is direct state transfer. For intermediate values of $d$, it is an HTDT protocol.}
    \label{Fig2}
\end{figure}

The role of the encoding operation is to partially correct the noise added by the communication channel $\mathcal{G}$. The encoding operation (labeled TMS in Fig.~\ref{Fig2}) is defined as  
\begin{equation}
   X_{\rm enc}=
    \begin{pmatrix}
        \sqrt{d}\mathbb{I}_2&\sqrt{d-1}\sigma_z\\
        \sqrt{d-1}\sigma_z&\sqrt{d}\mathbb{I}_2 \end{pmatrix} \;\;\;,\;\;\;Y_{\rm enc}=0\,,\label{encoding}
\end{equation}
where \(d\geq1\), and corresponds to a two-mode squeezing acting only on Alice’s mode. If $d$ is infinite, measuring the quadratures $\hat{x}$ and $\hat{p}$ on one of the outputs of the encoding operation is equivalent to a Bell measurement. Alice's encoded system consists of two modes with first moments $\sqrt{d}v_T$ and $\sqrt{d-1
}\sigma_z v_T$. We discard the complex conjugate mode (the one affected by $\sigma_z$), and transmit the other mode to Bob via the channel \eqref{phasechannel}.

Finally, as a decoding operation, Bob performs a beamsplitter operation on his part of the entangled mode and the received mode. This operation is defined by the matrices 
\begin{equation}
    X_{\rm dec}=
   \begin{pmatrix}
        \sqrt{\tau}\mathbb{I}_2&\sqrt{1-\tau}\mathbb{I}_2\\
        -\sqrt{1-\tau}\mathbb{I}_2&\sqrt{\tau}\mathbb{I}_2
        \end{pmatrix}
\;\;\;,\;\;\;Y_{\rm dec}=0\,,\label{decoding}
\end{equation}
with \(0<\tau\leq1\). It can be implemented, e.g., with a directional coupler in a microwave experiment~\cite{PhysRevLett.117.020502,Fedorov2021experimental, Abdo2025}.  One of the two outputs is
\begin{align}
&v_{\rm out}= \sqrt{g}v_{\rm T} \;\;\;,\;\;\;
    \Gamma_{\rm out}=g\Gamma_{\rm T}+G\mathbb{I}_2 \label{finaloutput}
\,,
\end{align}
where \(g=dx\tau\) is the total gain of the overall channel, and 
\begin{align}\label{G}
    G=&\left(g-\frac{g}{d}\right)a+\left(1-\frac{g}{dx}\right)b\nonumber\\\quad&-2c\sqrt{\left(1-\frac{g}{dx}\right)\left(g-\frac{g}{d}\right)}+\frac{gy}{dx}
\end{align}
is the total added noise. Notice that, since \({0<\tau\leq1}\), we must have \({0< g/dx\leq1}\) and, consequently, the squeezing parameter should satisfy the relation \({d\geq\max\{g/x,1\}}\).

In practice, the analog protocol {\it simulates} the action of a phase-insensitive Gaussian channel with gain $g$ and noise $G$ onto the state \(\rho_{\rm T}\). We point out that Eq.~\eqref{finaloutput} can be interpreted as follows: For a given channel gain $g$, which can be modulated via $d$ and $\tau$, the analog protocol will provide a noise $G$ given in \eqref{G}. This result can be generalized to the simulation of a generic non-degenerate Gaussian channel by exploiting the decomposition given in Eq.~\eqref{diagonalization}.

\subsection{Analog Protocol: A Hybrid Approach}\label{hybrid}

In the limit ${d\to \infty}$, we get ${G=ga+b-2\sqrt{g}c\equiv {G_{\rm qt}}}$, which matches the added noise in quantum teleportation for simulating a phase-insensitive channel with gain $g$~\cite{Adesso2017}. In this limit, it is an Analog Quantum Teleportation protocol, i.e., a quantum teleportation scheme that employs analog error correction to mitigate communication-channel losses. For ${d=1}$, it reduces to direct state transfer. For intermediate values of $d$, the protocol is a hybrid of teleportation and direct transmission (HTDT). To illustrate this, consider a quantum teleportation protocol that uses only Alice's discarded mode (see Fig.~\ref{Fig2}), where Alice performs a heterodyne measurement on this mode and sends the result to Bob, who then applies an appropriate unitary operation to realize an overall channel with gain $g$ (i.e., to \emph{simulate} a channel with gain $g$). The resulting noise $G_{\rm dis}$ is given by 
\begin{align}
G_{\rm dis} = \frac{gd}{d-1}a + b - 2c\sqrt{\frac{gd}{d-1}}+\frac{g}{d-1}.
\end{align}
For any choice of $g$ and triplets ${(a,b,c)}$, we have that ${G/G_{\rm dis}\to1}$ for ${d\to\infty}$, ${G/G_{\rm dis}\to 0}$ for ${d\to1}$, and ${G/G_{\rm dis}<1}$ for ${1<d<\infty}$.

In the limit ${d \to \infty}$, the discarded and transmitted modes become perfectly correlated. A heterodyne measurement on the discarded mode collapses Alice's transmitted mode onto a coherent state whose amplitude is equal to the complex conjugate of the measurement outcome multiplied by a known, large gain factor. This measurement also projects Bob’s mode onto the target state up to a displacement determined by the complex conjugate of the measurement outcome. Notice that since the measurement on the discarded mode acts as a Bell measurement on the unknown input state and Alice’s share of the entanglement resource, and its outcome is correlated with the transmitted mode, it constitutes a QND measurement in the Bell basis. 

When the discarded mode is measured, there are two ways to transmit the resulting information to Bob. In standard teleportation, one would transmit the measurement outcome through a classical channel, and the performance is characterized by $G_{\rm dis}$. In the analog protocol, instead, the collapsed transmitted mode itself is sent through the quantum channel, with overall noise ${G=G_{\rm qt}}$. Since ${G_{\rm dis} \to G_{\rm qt}}$ in the limit ${d \to \infty}$, the two approaches are equivalent. The only difference is the way classical information is transmitted: The analog protocol transmits the information using an amplified coherent state, with the amplification serving as analog error correction against the channel noise. 

It is essential to note that the analog scheme remains deterministic: in practice, no measurement on the discarded mode is actually required, and simply tracing it out is sufficient. However, the measurement-based picture is conceptually useful, as it clarifies why the protocol reduces to quantum teleportation in this limit. The underlying concept is referred to in the literature as ``all-optical quantum teleportation''~\cite{Ralph1999} or ``analog feedforward''~\cite{Dicandia2015}. Experimental implementations of this limit include Refs.~\cite{Liu2020Orbital, Liu2024} in the optical regime and Refs.~\cite{Fedorov2021experimental,yam2025quantum,Abdo2025} in the microwave regime.

For ${1 < d < \infty}$, the correlations between the discarded and transmitted modes are not perfect, so the measurement on the discarded mode only partially determines the state of the transmitted mode. The information obtained from measuring the discarded mode is less strongly correlated with Bob's state, resulting in a larger noise $G_{\rm dis}$. Here, ${G/G_{\rm dis}<1}$ indicates that the analog protocol is a hybrid of quantum teleportation and direct transmission. In the limit ${d\to1}$, we have that ${G/G_{\rm dis}\to0}$, and the protocol realizes purely direct transmission of the unknown state.

\subsection{Performance analyses}\label{main}
We aim to identify the regimes in which HTDT outperforms quantum teleportation, whether digital or analog, in simulating a phase-insensitive channel. Such performance is quantified by the amount of added noise $G$, i.e., the less noise, the better the simulation. 
Optimizing $G$ with respect to the triplet $(a,b,c)$ and the squeezing parameter $d$ is challenging. There are trivial examples, such as the ${g=x}$ case, for which ${G=G_{\rm qt}+(y-G_{\rm qt})/d}$. This expression is linear in $1/d$ and is therefore minimized by ${d=1}$ or ${d\to\infty}$ depending on whether $G_{\rm qt}$ is greater or smaller than $y$. In this example, the problem reduces to optimizing $G_{\rm qt}$. 

Another simple case is when there is no pre-shared entanglement between Alice and Bob. We fix Alice and Bob's states to be pure, i.e., $(a,b,c)=(1,1,0)$. The encoding operation corresponds to a quantum-limited phase-insensitive amplification of the state $\rho_{\rm T}$. The protocol is implemented by sending the amplified state directly through the communication channel, obtaining  
\begin{align}\label{gdir}
 G_{\rm ef}&=1+g+\left[\frac{y-(1+x)}{x}\right]\frac{g}{d}\,.  
\end{align}
This can be optimized over $d$ with the constraint ${d\geq\max\{g/x,1\}}$. Since \eqref{gdir} is linear in ${1/d}$, $G_{\rm ef}$ is minimized for ${d\to\infty}$ if ${y\geq1+x}$ (i.e., if the channel is entanglement-breaking), or by ${d=\max\{g/x,1\}}$ if ${y<1+x}$. This case represents direct state transfer enhanced by local encoding/decoding operations.
Notice that $G_{\rm ef}$ can be lower than $G_{\rm qt}$ if the transmissivity of the communication channel is large enough and the added noise is sufficiently small. In fact, while $G_{\rm qt}$ is limited by the entanglement, $G_{\rm ef}$ is limited by the communication channel parameters.

Let us now consider the case in which Alice and Bob share some entanglement, which we take as the resource state. For the moment, we do not consider how this entanglement is generated or distributed by Charlie. Given a fixed logarithmic negativity of the resource state, $E_{\mathcal{N}}(\rho_{\rm AB}) = 2r$, one can derive a necessary and sufficient condition for the figure of merit $G$ to be minimized at a finite value of $d$.

\begin{theo*}\label{theorem1}
Consider a Gaussian communication channel $\mathcal{G}$ with transmissivity ${x>0}$ and added noise $y$, as in described in Eq.~\eqref{phasechannel}, and assume that Alice and Bob share an entangled resource with logarithmic negativity $2r$.  HTDT  outperforms quantum teleportation—optimized over the set of resource entangled states—in the task of simulating Gaussian phase-insensitive channels with gain ${g\in[\tanh(r),\coth(r)]}$, if and only if ${y< e^{-2r}\left(1+x\right)}$.
\end{theo*}
\begin{proof}
An optimization of $G_{\rm qt}$, considering as entanglement resource a state with logarithmic negativity $2r$, has been found in \cite{Adesso2017}. We first set ${c=\sqrt{(a-e^{-2r})(b-e^{-2r})}}$ in order to have ${E_{\mathcal{N}}(\rho_{\rm AB})=2r}$, obtaining
\begin{align}
G_{\rm qt} = \left[(ga-ge^{-2r})^{\frac{1}{2}}-(b-e^{-2r})^{\frac{1}{2}}\right]^2+G_{\rm qt}^*,
\end{align}
where $G_{\rm qt}^*=e^{-2r}(1+g)$. Then, we fix $a$ such that $G_{\rm qt}=G_{\rm qt}^*$, which is the boundary of the channels accessible by quantum teleportation.  Overall, we get the triplet
\begin{align}\label{entanglementtriplet}
    a&=\frac{b+e^{-2r}(g-1)}{g}\;\;\;,\;\;\; c= \frac{b-e^{-2r}}{\sqrt{g}}\,,\nonumber\\
    b&\geq\frac{e^{2r}g+e^{-2r}-|g-1|}{g+1-e^{2r}|g-1|}\equiv b^*\,,
\end{align}
where ${b\geq b^*}$ is needed to satisfy the constraint that \(\rho_{\rm AB}\) is physical, and the solution holds when ${\tanh(r)\leq g\leq\coth(r)}$~\cite{Adesso2017}. Notice that ${b=b^*}$ corresponds to the optimal $\rho_{\rm AB}$ with minimal energy.

Let us now consider the noise $G$ in \eqref{G} with generic parameters. On the one hand, if we set the logarithmic negativity of $\rho_{\rm AB}$ to $2r$, Eq.~\eqref{G} can be written as
\begin{align}
G=&\left[\left(g-\frac{g}{d}\right)^{\frac{1}{2}}(a-e^{-2r})^{\frac{1}{2}}-\left(1-\frac{g}{dx}\right)^{\frac{1}{2}}(b-e^{-2r})^{\frac{1}{2}}\right]^2\nonumber\\\quad&+\frac{g}{dx}\left[y-e^{-2r}(1+x)\right]+G_{\rm qt}^*.
\end{align}
Therefore, if ${y\geq e^{-2r}(1+x)}$ we have $G\geq G_{\rm qt}^*$ for any choice of $a$, $b$, and $d$, and quantum teleportation with the set of parameters in \eqref{entanglementtriplet} is optimal.

On the other hand, we can expand $G$ to the first order in $1/d$ for $d\to\infty$, obtaining
\begin{align}
    G&=G_{\rm qt}+\frac{G'}{d}+o\left(\frac{1}{d}\right)\,,\label{G1}\\\quad
    G'&=-g\left(a+\frac{b}{x}\right)+\frac{c\sqrt{g}(g+x)}{x}+\frac{gy}{x}\,.\label{Gprime}
\end{align}
The condition ${G'<0}$ computed using any triplet in Eq.~\eqref{entanglementtriplet} is sufficient for ${G<G_{\rm qt}^*}$ to hold for a finite $d$. Such condition is equivalent to ${y<e^{-2r}(1+x)}$. 
\end{proof}

Notice that $G'<0$ in Eq.~\eqref{Gprime} gives a sufficient condition for the analog protocol to be optimal for any fixed $a$, $b$, and $c$, not necessarily optimized for teleportation. This condition is equivalent to ${y<ax+b-c(g+x)/\sqrt{g}}$.

This result can be generalized to situations where the communication channel is described by Eq.~\eqref{diagonalization} that have a canonical form given in \eqref{phasechannel}. It depends uniquely on the structure of the communication channel, and establishes that a necessary and sufficient condition for quantum teleportation to be optimal in the task of simulating Gaussian channels is that the communication channel reduces entanglement when applied to a part of {\it any} bipartite entangled state with logarithmic negativity $2r$. The result can be rephrased as addressing the following question: Given a communication channel that cannot be simulated with quantum teleportation, can it be used to implement a target channel that is also not simulatable with teleportation and differs from the original communication channel?

\section{Optimal quantum state transfer}\label{opt_state_main}

\subsection{Gaussian Codebooks of Coherent States}
We consider quantum state transfer of a Gaussian distributed codebook of coherent states $|\alpha\rangle$.  The fidelity averaged on the codebook can be computed as $\overbar{\mathcal{F}}=\int_{\mathbb{C}}\frac{\lambda}{\pi}e^{-\lambda |\alpha|^2}\langle \alpha | \rho_{\rm out} | \alpha\rangle d^2\alpha$ with $\lambda>0$. Using Eq.~\eqref{finaloutput}, we obtain ${\overbar{\mathcal{F}}=2\lambda/[2(1-\sqrt{g})^2+\lambda(1+g+G)]}$~\cite{Adesso2017}. This expression shall be optimized with respect to the gain $g$ and the overall noise $G$. However, it is clear that for a fixed $g$, a lower $G$ is beneficial. For the sake of the discussion, let us consider the ${\lambda\to0}$ limit, i.e., the limit of a uniformly distributed codebook. In this case, $\overbar{\mathcal{F}}$ is zero for ${g\neq1}$, while it is ${2/(2+G)}$ for ${g=1}$. 

In the following, we focus on the case of the quantum-limited attenuator as a communication channel, defined by Eq.~\eqref{phasechannel} with ${0< x\leq1}$ and ${y=1-x}$. This channel plays a crucial role in, e.g., superconducting microwave technology, where the communication can be implemented with a cryogenic link~\cite{walraff2020, yam2023}. 

\subsection{Noiseless Entanglement Distribution}\label{opt_state}

We first considered the case in which the entanglement distribution channel is different from the communication channel. In general, these channels need not coincide, as the entanglement distribution channel may rely on dedicated infrastructure and therefore be less noisy. For instance, in an optical implementation, entanglement may be distributed through optical fibers, where power loss rates can be as low as $10^{-4}$~dB/m~\cite{petrovich2025broadband}, while communication occurs in free space, resulting in substantially larger loss rates. In this scenario, the shared entanglement between Alice and Bob is only mildly affected by losses, whereas the communication channel can experience significant attenuation. 

We set ${g=1}$ and consider the triplet optimizing $G_{\rm qt}$ with minimal energy, i.e., $(a,b,c)=(\cosh(2r),\cosh(2r),\sinh(2r))$, see Eq.~\eqref{entanglementtriplet}. Note that, while choosing a non-pure state satisfying Eq.~\eqref{entanglementtriplet} (i.e., ${a = b > \cosh(2r)}$) leaves $G_{\rm qt}$ unchanged, this choice affects the performance of the analog protocol in the regime where a finite $d$ is optimal, in the sense that the entanglement-free protocol can become optimal for finite transmissivities.
Nevertheless, the region where the analog protocol with finite $d$ encoding is optimal is still given by the Theorem in Section~\ref{main}.

\begin{figure}[t!]
    \centering
\includegraphics[width=0.47\textwidth]{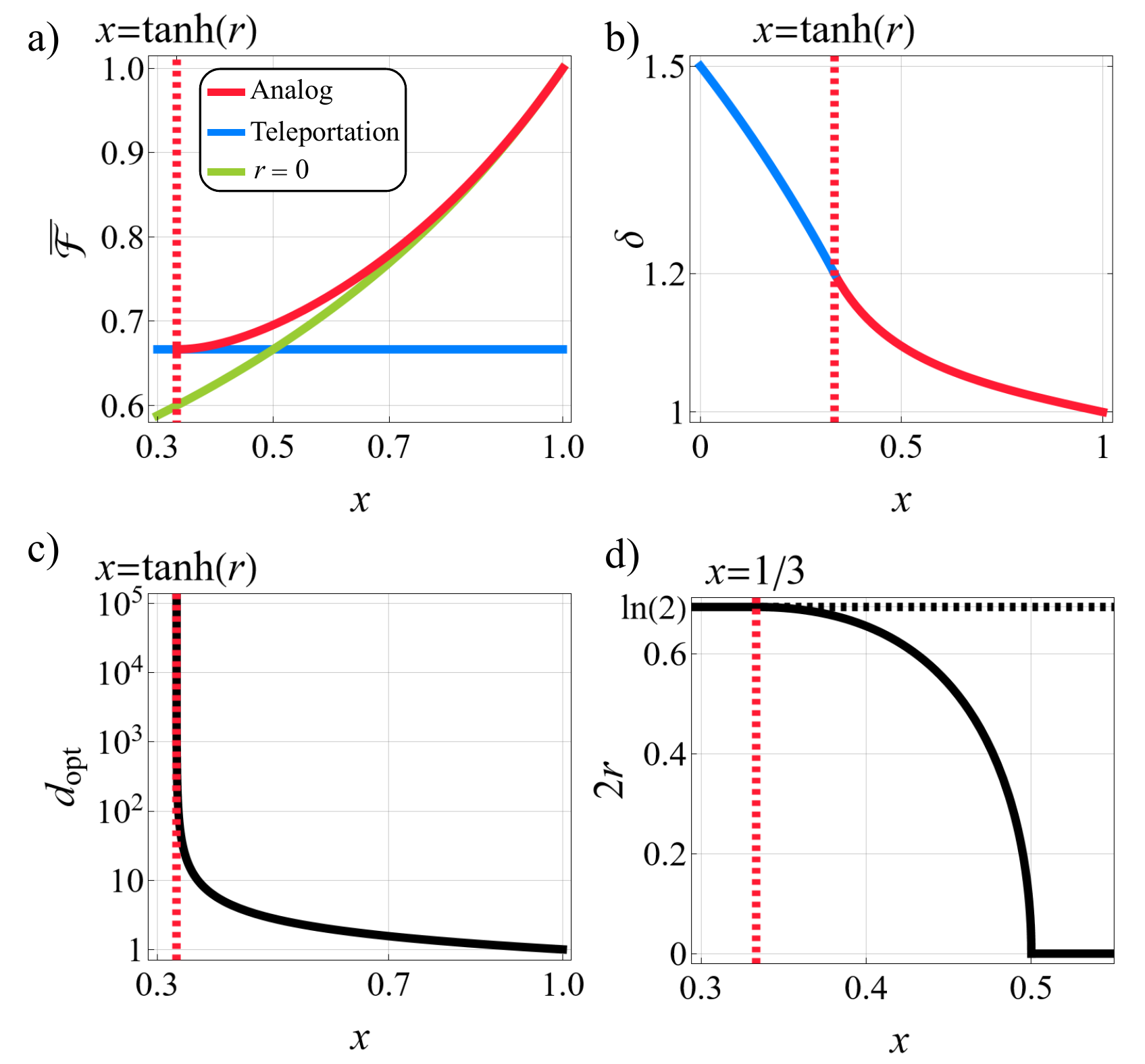}
    \caption{\textbf{Average fidelities, optimal encoding, and no-cloning threshold.} \textbf{a)} Average fidelity $\overbar{\mathcal{F}}$ of HTDT optimized over $d$ (red), quantum teleportation (blue), and entanglement-free case (green), for the state transfer of a codebook of uniformly distributed coherent states. We consider a quantum-limited attenuator with transmissivity $x$ as a communication channel, and 
 the optimal triplet for quantum teleportation ${(a,b,c)=(\cosh(2r),\cosh(2r),\sinh(2r))}$ as resource. In {\bf b)}, we plot the infidelities ratio \({\delta=(1-\overbar{\mathcal{F}}_{\rm ef})/(1-\overbar{\mathcal{F}}_{\rm an})}\), showing the advantage of the quantum teleportation (blue) and HTDT (red) over the entanglement-free case. {\bf c)} Plot of the optimal $d$ for HTDT, i.e., \({d_{\rm opt}=[x-\tanh^2(r)
]/[x^2-\tanh^2(r)]}\). In {\bf a)} and {\bf c)} we set ${2r=\ln(2)\simeq 0.693}$, such that ${\overbar{\mathcal{F}}_{\rm qt}=2/3}$ matches the no-cloning threshold~\cite{Cerf2000}. Changing $r$ results in qualitatively similar plots. {\bf d)} Plot of the logarithmic negativity $2r$ needed to reach the no-cloning threshold for the optimal analog protocol (continuous) and quantum teleportation (black dashed) protocols.}
    \label{Fig3}
\end{figure}

The resulting average fidelity for quantum teleportation is ${\overbar{\mathcal{F}}_{\rm qt}=1/(1+e^{-2r})}$. For the analog protocol optimized over $d$, we obtain ${\overbar{\mathcal{F}}_{\rm an}=\overbar{\mathcal{F}}_{\rm qt}}$ for ${x\leq \tanh(r)}$ and ${\overbar{\mathcal{F}}_{\rm an}= 2x/[1+x(2-x)-(1-x)^2\cosh(2r)]}$ for  ${x>\tanh(r)}$. This shall be compared with the average fidelity of the entanglement-free case, i.e., ${\overbar{\mathcal{F}}_{\rm ef}=1/(2-x)}$. Notice that ${\overbar{\mathcal{F}}_{\rm ef}\leq \overbar{\mathcal{F}}_{\rm an}}$ always holds, and the equality is achieved at ${x=1}$. However, $\overbar{\mathcal{F}}_{\rm ef}$ is larger than $\overbar{\mathcal{F}}_{\rm qt}$ as soon as ${x>1-e^{-2r}}$. 

We plot these fidelities in Fig.~\ref{Fig3} for ${2r=\ln(2)}$, which makes quantum teleportation reach the no-cloning threshold ${\overbar{\mathcal{F}}_{\rm qt}=2/3}$~\cite{Cerf2000}. Fig.~\ref{Fig3}(a) illustrates how the analog protocol with finite $d$ (i.e., HTDT) outperforms quantum teleportation for ${x>\tanh(r)}$, while Fig.~\ref{Fig3}(c) plots the optimal value of $d$.  In Fig.~\ref{Fig3}(d), we show the entanglement required by our analog protocol to reach the no-cloning threshold. Here, we identify the ${1/3<x<1/2}$ region, where the needed entanglement is ${2r=2\arcsinh\left[\sqrt{x(1-2x)/2(1-x)^2}\right]}$, as transitional to the region where entanglement is no longer necessary.

\subsection{Noise in Entanglement Distribution}\label{distribution} 
We now discuss the effect of noise and dissipation on entanglement distribution. As in the previous section, we focus on the task of state transfer of a uniformly distributed codebook of coherent states, where having $g=1$ and a symmetric entangled state between Alice and Bob is optimal. We assume that both the distribution and communication channels are quantum-limited attenuators (or pure-loss channels) with the same power loss rate $\gamma$ (expressed in dB/m), so that the power transmission between nodes X and Y at distance $D_{\rm XY}$ (expressed in m) is ${x_{\rm XY}=10^{-\frac{\gamma}{10}D_{\rm XY}}}$, see Fig.~\ref{Fig4}. In this setting, the triangle inequality is equivalent to $x_{\rm AB}\geq x_{\rm CA}x_{\rm CB}$. We consider Charlie to be equidistant from Alice and Bob, so that we have a symmetric entangled state at the level of Alice and Bob. More general configurations require specialized treatment, where an optimization over Charlie's resources is performed.

Let us consider Charlie to generate a pure state with entanglement $2 r_{\rm C}$. The logarithmic negativity between Alice and Bob is obtained by computing the smallest symplectic eigenvalue of the state distributed from Charlie to Alice and Bob through the attenuator channels, obtaining ${2r=-\log\left(1-x_{\rm C}+x_{\rm C}e^{-2r_{\rm C}}\right)}$,  where ${x_{\rm CA}= x_{\rm CB}\equiv x_{\rm C}}$. Notice that the entanglement saturates at ${2r=-\log\left(1-x_{\rm C}\right)}$ for ${2r_{\rm C}\gg1}$. 

\begin{figure}[t!]
    \centering
    \includegraphics[width=0.35\textwidth]{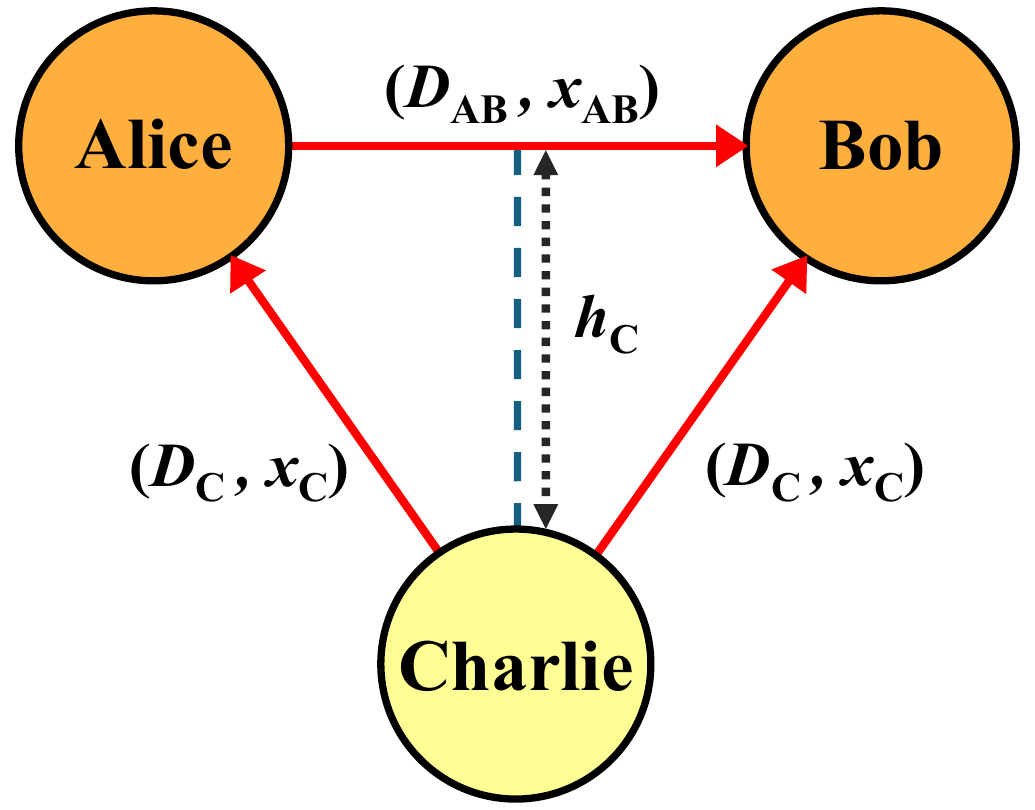}
    \caption{{\bf Alice–Bob–Charlie configuration.} We consider a configuration in which Charlie is equidistant from Alice and Bob, i.e., ${D_{\rm CA} = D_{\rm CB} \equiv D_{\rm C}}$. All channels have the same loss rate $\gamma$, so that ${x_{\rm CA}=x_{\rm CB}\equiv x_{\rm C}}$. The quantity $h_{\rm C}$ denotes Charlie’s distance from the Alice–Bob axis.}
    \label{Fig4}
\end{figure}
Using the Theorem derived in Section~\ref{main}, we obtain the condition ${2r_{\rm C}<-\log\left(1-\frac{2x_{\rm AB}}{x_{\rm C}(1+x_{\rm AB})}\right)}$ for an analog protocol with {\it finite} encoding squeezing (i.e., HTDT) to outperform quantum teleportation. Since ${x_{\rm C}\leq \sqrt{x_{\rm AB}}}$ always holds, we have that for ${2r_{\rm C}<-\log\left(1-\frac{2\sqrt{x_{\rm AB}}}{1+x_{\rm AB}}\right)}$ an HTDT protocol is optimal for any position of Charlie.

Let us now consider typical values of the attenuation coefficient and entanglement in realistic platforms. In superconducting microwave technology, niobium–titanium coaxial cables exhibit power losses as low as ${\gamma=10^{-3}}$~dB/m at frequencies around $5$~GHz~\cite{Fedorov2021experimental,yam2025quantum}. While we focus on the microwave superconducting case, a similar analysis can be carried out for optical fibers, which can exhibit even lower loss rates, ${\gamma \simeq 10^{-4}}$~dB/m~\cite{petrovich2025broadband}, allowing for significantly longer transmission distances. In Fig.~\ref{Fig5}, we plot the average fidelities for the optimized analog protocol, quantum teleportation, and the entanglement-free case for ${D_{\rm AB}\simeq1.5}$~km. Notice that, since the entangled state shared by Alice and Bob is {\it not} pure, a protocol that uses no entanglement at all can outperform the analog protocol that relies on entanglement. There is a wide range of entanglement values for which HTDT is optimal. Of particular interest is the case ${2r_{\rm C} \simeq 2.1}$, corresponding to the logarithmic negativity reported in recent microwave experiments~\footnote{In Ref.~\cite{Abdo2025}, the logarithmic negativity is defined using a logarithm in base~2 rather than the natural logarithm. Converting to our convention, they report a logarithmic negativity between Alice and Bob of \({2r \simeq 1.04}\), with the entanglement distributed asymmetrically. At the source, the state is symmetric and exhibits a maximum squeezing below the vacuum level for the joint quadrature of \(9\,\mathrm{dB}\), which corresponds to a logarithmic negativity at Charlie of \(2r_{\mathrm{C}} \simeq 2.07\).}. Here, HTDT still significantly outperforms teleportation, and the optimal amplification is \(d \simeq 3.1\) for \(h_{\mathrm{C}} \simeq 0~\mathrm{m}\) (Charlie located exactly midway between Alice and Bob). This value is experimentally achievable without introducing significant additional noise~\cite{menzel2012path,zhong2013squeezing,huo2018deterministic}. For larger values of $2r_{\rm C}$, the advantage of HTDT over quantum teleportation (${d\to \infty}$) gradually diminishes.

\begin{figure}[t!]
    \centering
    \includegraphics[width=0.48\textwidth]{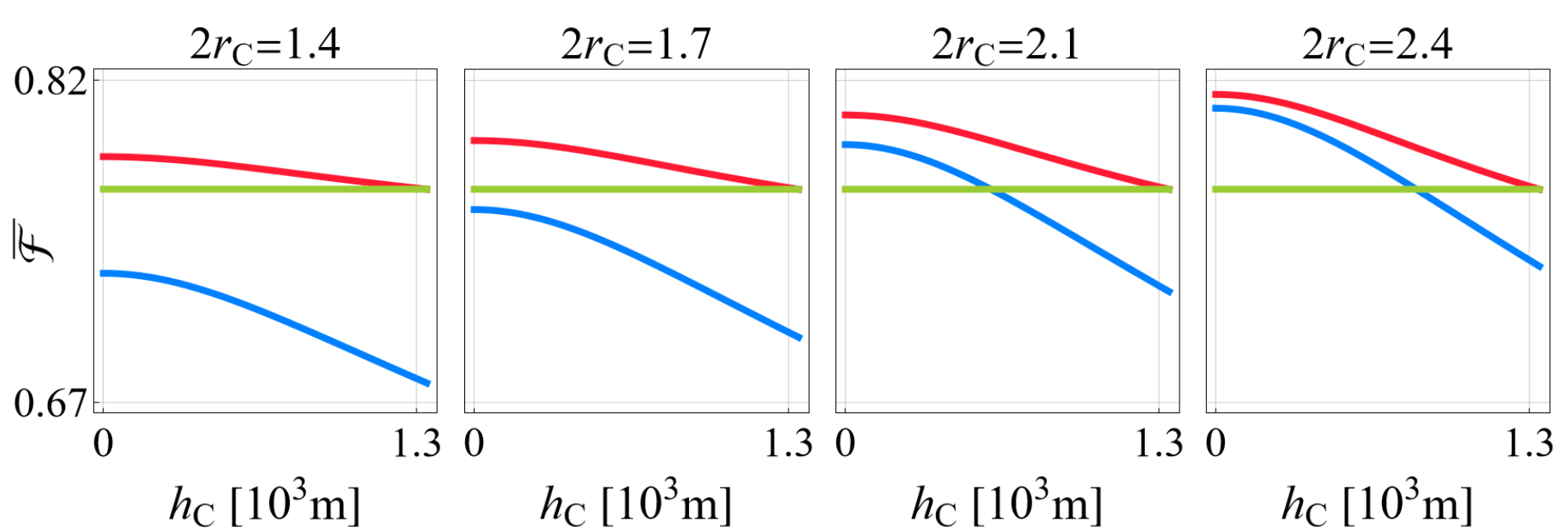}
    \caption{{\bf Average fidelity versus Charlie’s position.} We plot the average fidelities for the analog protocol optimized over the encoding squeezing (red), quantum teleportation (blue), and the entanglement-free case (green) for various values of Charlie's logarithmic negativity ($2r_{\rm C}$). The results are shown for the configuration in Fig.~\ref{Fig4} ($D_{\rm CA} = D_{\rm CB} \equiv D_{\rm C}$) as a function of Charlie's vertical position ($h_{\rm C}$), for parameter values ${\gamma=10^{-3}}$~dB/m, ${x_{\rm AB}=0.7}$ (corresponding to ${D_{\rm AB}\simeq 1.5}$~km), and a symmetric pure state at Charlie. The red curves are optimized for a {\it finite} encoding squeezing, so HTDT always outperforms quantum teleportation in the shown regimes.}
    \label{Fig5}
\end{figure}
\section{Discussion}\label{discussion}

In this work, we studied an analog protocol for simulating non-degenerate Gaussian channels. The protocol interpolates between quantum teleportation and direct transmission, realizing a hybrid teleportation–direct transmission (HTDT) scheme for intermediate values of the encoding parameter. We derived a necessary and sufficient condition under which this hybrid setup is optimal (see Theorem in Section~\ref{main}).
Our result is relevant as it mitigates noise in cases where entanglement resources and/or losses in the communication channel are limited. In this context, we have demonstrated that an HTDT protocol can enhance the average fidelity achievable with quantum teleportation. This has been proved both in the case when the entanglement distribution channel is virtually noiseless (see Section~\ref{opt_state}), and in the  ``fair'' case where all channels have the same loss rate (see Section~\ref{distribution}). We have shown that for distances between Alice, Bob, and Charlie on the order of kilometers, HTDT outperforms quantum teleportation using experimentally achievable parameters. 

Taken together, these results demonstrate the overall practical relevance of the proposed hybrid protocol, even under equivalent channel conditions. They directly connect our analysis to existing microwave and optical quantum communication platforms, where analog feedforward has already been demonstrated experimentally~\cite{Fedorov2021experimental,yam2025quantum,Abdo2025,Liu2020Orbital,Liu2024}. In particular, our treatment explicitly includes superconducting microwave implementations, where cryogenic links operating at $20$~mK~\cite{walraff2020,yam2023} realize quantum-limited attenuating channels. Such cryogenic interconnects are a key technological element in modular quantum computing architectures based on superconducting qubits, making our framework naturally applicable in this context.

Finally, although this work has focused on continuous-variable systems, the result may also be extended to discrete-variable settings. The encoding operation in our analog protocol can be viewed as a quantum non-demolition (QND) measurement in the Bell basis~\cite{Radim2009,Radim2010}, see Section~\ref{hybrid}. In the discrete regime, a comparable QND measurement in the Bell basis, appropriately optimized for the specific noise characteristics, should be implemented on Alice’s side. 

{\it Acknowledgments---} We thank Radim Filip for discussions on the all-optical quantum teleportation~\cite{Ralph1999} and a potential discrete variable extension of the analog protocol. U.A. and R.D. acknowledge financial support from the Academy of Finland, grants no. 349199, 353832, and 368477. S.F. acknowledges financial support from PNRR MUR project PE0000023-NQSTI financed by the European Union – Next Generation EU, and from CQSense project financed by Fondazione Compagnia di San Paolo.

\bibliography{bibliography.bib}

\end{document}